%% file: root.tex
\documentclass[letterpaper, 10 pt, conference]{ieeeconf} 

\IEEEoverridecommandlockouts  
\usepackage{amsmath}
\usepackage{amssymb}
\usepackage{amsthm}
\usepackage{mathtools}
\usepackage{derivative}
\usepackage{xcolor}
\usepackage{bm}
\usepackage{mathrsfs}

\usepackage[noadjust]{cite}
\usepackage{url}
\usepackage[hidelinks]{hyperref}

\usepackage{graphicx}
\usepackage[export]{adjustbox} 
\graphicspath{{figures/}}
\usepackage{stfloats}

\newtheorem{theorem}{Theorem}
\newtheorem{lemma}{Lemma}

\newtheorem{corollary}{Corollary}

\theoremstyle{definition}
\newtheorem{definition}{Definition}

\newtheorem{remark}{Remark}

\DeclareMathOperator*{\argmin}{arg\,min}

\title{\LARGE \bf Control Barrier Function Synthesis for \\ 
Nonlinear Systems with Dual Relative Degree
}

\author{Gilbert Bahati, Ryan K. Cosner, Max H. Cohen, Ryan M. Bena, and Aaron D. Ames
\thanks{This work is supported by BP and by TII under project \#A6847.}
\thanks{The authors are with the Department of Mechanical and Civil Engineering, California Institute of Technology, Pasadena, CA $\{$\texttt{gbahati}, \texttt{rkcosner}, \texttt{maxcohen}, \texttt{ryanbena}, \texttt{ames}$\}$\texttt{@caltech.edu}
.}
}

\begin{document}

\maketitle

\input{def}

\thispagestyle{empty}
\pagestyle{empty}

\begin{abstract}
Control barrier functions (CBFs) are a powerful tool for synthesizing safe control actions; however, constructing CBFs remains difficult for general nonlinear systems.
In this work, we provide a constructive framework for synthesizing  CBFs for systems with \textit{dual relative degree}---where different inputs influence the outputs at two different orders of differentiation; this is common in systems with orientation-based actuation, 
such as unicycles and quadrotors. In particular, we propose \textit{dual relative degree CBFs (DRD-CBFs)} and show that these DRD-CBFs can be constructively synthesized and used to guarantee system safety. 
Our method constructs DRD-CBFs by leveraging the dual relative degree property---combining a CBF for an integrator chain with a Lyapunov function certifying the tracking of safe inputs generated for this linear system.
We apply these results to dual relative degree systems, both in simulation and experimentally on hardware using quadruped and quadrotor robotic platforms.
\end{abstract}

\section{Introduction}

Control invariance has become a powerful descriptor of safety requirements in modern control systems, where tools such as reachability \cite{BansalCDC17,TomlinTAC21}, model-predictive control (MPC) \cite{borrelli2017predictive,WabersichTAC23}, and control barrier functions (CBFs) \cite{AmesTAC17} provide a framework for synthesizing safety-critical controllers. 
These approaches come with various trade-offs that may lead to benefits/drawbacks in different situations.
One benefit of CBFs over other techniques is their efficient online computation;
however, generating valid CBFs -- those consistent with the system dynamics -- from user-defined safety requirements (e.g., position constraints) is a challenging problem. Recently, various methods have emerged that enable the systematic synthesis of CBFs for relevant classes of systems, such as chains of integrators \cite{MurrayACC20}, strict feedback systems \cite{AndrewCDC22}, and partially feedback linearizable systems \cite{cohen2024constructive}. While effective in certain cases, each of these approaches imposes certain structural requirements on the system dynamics, which may not hold for systems of interest. Methods such as exponential \cite{SreenathACC16} and high order CBFs \cite{WeiTAC22} require less structure of the dynamics, but place more restrictive requirements on the safety constraints \cite{CohenARC24}.

\begin{figure}[t!]
    \centering
    \includegraphics[width=1.\linewidth]{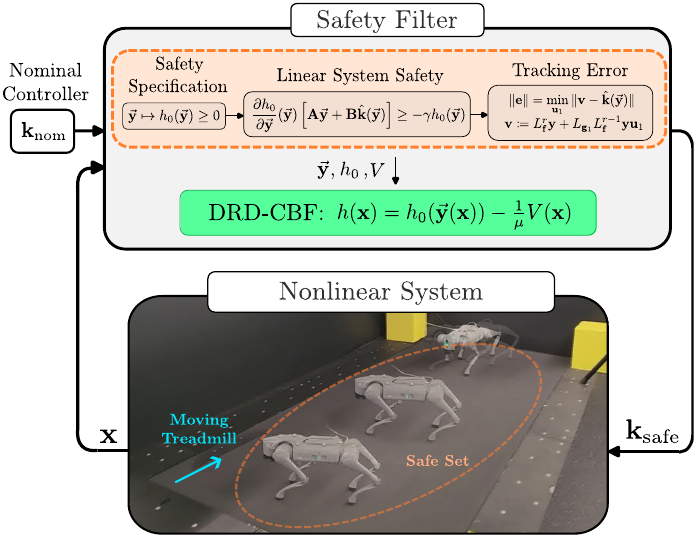}
    \vspace{-6mm}
    \caption{A nominal control block diagram for enforcing safety constraints on a nonlinear system via DRD-CBFs. Hardware experiments can be viewed at {\color{blue}\url{https://www.youtube.com/watch?v=l41vZb4WNqA}}.
    \label{fig:hero_figure}}
    \vspace{-6mm}
\end{figure}

Another important class of methods for CBF synthesis are those leveraging reduced-order models (ROMs) -- simplified representations of the original system -- in which safe commands generated by ROMs are tracked by more complex full-order models \cite{TamasRAL22,CohenARC24}. These methods generally avoid the difficulties in CBF synthesis by focusing on dramatically simplified systems, e.g., by modeling a quadrupedal robot as a unicycle, and achieve system-level safety through the convergence of the full-order model to the ROM with 
sufficient tracking rates. Thus, safe-set synthesis can be performed on a simplified system and then extended to the full-order system without directly considering higher-order dynamics. While these synthesis methods are often conservative, they are capable of constructing control invariant subsets of user-defined safety requirements for rather complex systems. 

User-defined safety requirements are often (sometimes implicitly) expressed through a choice of desired outputs, such as a system's position in Euclidean space \cite{cohen2024constructive}. However, these outputs and their Lie derivatives may not fully capture all the system states (e.g., orientation), limiting the ability to control the entire system effectively. This often arises when the outputs do not possess a valid relative degree with respect to all the control inputs. In such a situation, it is often possible to dynamically extend the system to include higher derivatives of the inputs until a valid relative degree is achieved \cite{Isidori}, and methods such as integral CBFs \cite{ames2020integral} may be leveraged for safety-critical control of dynamically extended systems \cite{WeiACC22}.
However, the resulting integral-based controllers typically introduce delays to the nonlinear system.
Dynamic extension is also closely related to techniques based on differential flatness \cite{fliess1999lie,murray1995differential}, which have demonstrated success in controlling highly dynamic systems such as quadrotors \cite{lee_geometric_2010,MellingerICRA2011}. Although there exists CBF synthesis methods tailored to specific differentially flat systems and safety constraints \cite{wu2016safety,SreenathACC16-quad}, a general characterization of these ideas remains undeveloped.

Inspired by the above ideas, we propose a new method of synthesizing CBFs for a special, but relevant, class of systems -- those with \textit{dual relative degree}, where different components of the input influence the outputs at two different orders of differentiation---capturing systems such as unicycles and quadrotors.
Compared to works like \cite{cohen2024constructive},
this relaxes the requirement that the outputs must have a relative degree in the traditional sense. Rather than dynamically extending our system to achieve a valid relative degree and design a safe controller, we use a ROM-inspired approach \cite{TamasRAL22} by leveraging a Lyapunov function to certify tracking of a class of linear systems by the nonlinear dynamics. The original safety constraint and this Lyapunov function enable the synthesis of a CBF for the full nonlinear system, yielding a \emph{dual relative degree} CBF (DRD-CBF). From this perspective, our work can be seen as extending methods tailored for specific differentially flat systems and safety constraints \cite{wu2016safety,SreenathACC16-quad} to a wider class of systems and CBFs. 

The main contributions of this work are three-fold.
First, we provide a characterization of systems with \textit{dual relative degree}.
Second, we provide a constructive framework for synthesizing CBFs for these systems. 
Third, we illustrate the utility of our approach with in-depth case studies, including hardware demonstrations on quadrupeds and quadrotors. 

\section{Technical Background}

\subsection{Control Barrier Functions}
Consider\footnote{The Euclidean norm is denoted as $\Vert \cdot \Vert$. We denote that $\alpha$ is in the class of extended class-$\mathcal{K}$ infinity functions as  $\alpha \in \mathcal{K}_{\infty}^e$. For a full column/row rank matrix $\mb{A}$, we use $\mb{A}^{\dagger}$ to denote the left/right Moore-Penrose pseudoinverse. With an abuse of terminology, we refer to a function as smooth if it is differentiable as many times as necessary.} 
a nonlinear control-affine system of the form: 
\begin{align}
    \dot{\mb{x}} = \mb{f}(\mb{x}) + \mb{g}(\mb{x}) \mb{u}, \label{eq:ol-sys}
\end{align}
where $\mb{x} \in \R^n $ is the system state, $\mb{u}\in \R^m$ is the input, $\mb{f}: \R^n \to \R^n$ is the drift dynamics, and $\mb{g}:\R^n \to \R^{n\times m}$ is the actuation matrix. Using a state-feedback controller $\mb{k}: \R^n \to \R^m$, one obtains the closed-loop system: 
\begin{align}
    \dot{\mb{x}} =\mb{f}_\textup{cl}(\mb{x}) =  \mb{f}(\mb{x}) + \mb{g}(\mb{x})\mb{k}(\mb{x}) \label{eq:cl-sys}.
\end{align}
When the closed-loop dynamics are locally Lipschitz, they admit a unique continuously differentiable solution $t\mapsto \bx(t)$ for any given initial state $\bx_0\in\R^n$, which, for ease of exposition, we assume exists for all $t\geq0$.
 In this paper, we formalize the notion of safety using the concept of forward invariance. In particular, we consider safety requirements characterized by the forward invariance of some user-defined set $\mathcal{C}\subset\R^n$ given as the $0$-superlevel set of some continuously differentiable function $h: \R^n \to \R$:
\begin{align}
    \mathcal{C} \coloneqq & \{ \mb{x} \in \R^n ~|~ h(\mathbf{x}) \geq 0  \}.\label{eq:safe-set}
\end{align}
CBFs are one tool to synthesize control actions that enforce the forward invariance (i.e., safety) of this set $\mathcal{C}$.

\begin{definition}[Control Barrier Function (CBF) \cite{AmesTAC17}]\label{def:cbf}
    A continuously differentiable function $h: \R^n \to \R$ defining a set $\mathcal{C}\subset\R^n$ as in \eqref{eq:safe-set} is a \textit{control barrier function} (CBF) for \eqref{eq:ol-sys} if there exists an $\alpha \in \mathcal{K}^e_\infty$ such that for all $\mb{x}\in \R^n: $
    \begin{align}\label{eq:cbf_constraint}
      \sup_{\bu \in \real^m} \Bigg [ \ \underbrace{\frac{\partial  h}{\partial \bx}(\bx)\bf(\bx)}_{L_\bf h(\bx)} + \underbrace{\frac{\partial  h}{\partial \bx}(\bx)\bg(\bx)}_{L_\bg h(\bx)} \bu     \Bigg ] > -\alpha(h(\bx)).      
  \end{align}
\end{definition}

A CBF induces a point-wise set of safe inputs: 
\begin{equation*}
    \mathscr{K}_\textup{CBF}(\mb{x}) = \{ \mb{u} \in \R^m | L_\mb{f}h(\mb{x}) + L_\mb{g}h(\mb{x}) \mb{u} \geq - \alpha (h(\mb{x}))\}, 
\end{equation*}
such that any locally Lipschitz controller $\bk$ satisfying $\bk(\bx)\in\mathscr{K}_\textup{CBF}(\mb{x})$ enforces forward invariance of $\mathcal{C}$ \cite{AmesTAC17}.
Given a nominal controller $\bk_{\mathrm{nom}}:\real^n \rightarrow \real^m$, the following quadratic programming-based control law \textit{filters} the nominal controller by minimally adjusting $\bk_{\mathrm{nom}}$ to find the nearest safe action:
\begin{align}
    \bk(\bx) = &\argmin_{\bu \in \real^m} &&\|\bu - \bk_{\mathrm{nom}}(\bx)\|^2 \tag{CBF-QP} \label{eq:cbfqp}\\
    & \quad 
 \mathrm{s.t.} &&L_\bf h(\bx) + L_\bg h(\bx)\bu  \geq - \alpha(h(\bx)). \nonumber
\end{align}
While the above controller guarantees safety under ideal circumstances, it is often useful to robustify such controllers to unexpected disturbances. One tool for addressing this issue is an \emph{input-to-state-safe  CBF} (ISSf-CBF)\footnote{We restrict $\alpha$ to be a linear function for the remainder of
the paper.}. 
\begin{definition}[ISSf-CBF \cite{AnilLCSS22}]
    A continuously differentiable function $h\,:\,\R^n\rightarrow\R$ is said to be an \textit{input-to-state safe control barrier function } (ISSf-CBF) for \eqref{eq:ol-sys} on $\mathcal{C}$ as in \eqref{eq:safe-set} if there exists $\gamma>0$ and $\varepsilon > 0$ such that for all $\mb{x} \in \R^n $: 
    \begin{align}
        \sup_{\mb{u}\in \R^m}\left[ L_\mb{f}h(\mb{x}) + L_\mb{g}h(\mb{x}) \mb{u} \right] >  - \gamma h(\mb{x}) + \frac{\Vert L_\mb{g}h(\mb{x})\Vert^2}{\varepsilon}. \label{eq:issf}
    \end{align}
\end{definition}
ISSf-CBFs include a robustness margin 
$\tfrac{1}{\varepsilon}\|L_{\bg}h(\bx)\|^2$ to mitigate the impact of disturbances while providing practical safety guarantees \cite{AnilLCSS22}. Note that if $h$ satisfies \eqref{eq:cbf_constraint}, then it also satisfies \eqref{eq:issf} as robustness is only added when there is control actuation available (i.e., when $\Vert L_\mb{g}h(\bx) \Vert\neq\mathbf{0}$). Thus, \eqref{eq:issf} increases the robustness of safety to disturbances while ensuring that $\mathscr{K}_{\textup{CBF}}(\mb{x})$ remains nonempty.

\subsection{Outputs and Coordinate Transformations}
While CBFs provide a powerful approach for synthesizing safety-critical controllers, their success relies on knowledge of a function $h$ satisfying Def. \ref{def:cbf}. 
In general, constructing CBFs for nonlinear systems can be mapped to a
backwards reachability problem \cite{ChoiCDC21}; however, when the dynamics satisfy certain structural properties, the synthesis of CBFs can be made systematic \cite{AndrewCDC22,cohen2024constructive}. This structure may be revealed by selecting a set of outputs with a relative degree.
\begin{definition}[Relative Degree $r$ \cite{Isidori}]\label{def:relative-degree}
    A smooth function $\by\,:\,\R^n\rightarrow\R^p$ has \emph{relative degree} $r\in\mathbb{N}$ for \eqref{eq:ol-sys} if:
    \begin{align}
        L_{\bg}L_{\bf}^{i}\by(\bx) \equiv \mathbf{0}, &\quad \forall i\in\{0,\dots,r-2\},\label{eq:zero-lie-derivatives} \\
        \mathrm{rank}(L_{\bg}L_{\bf}^{r-1}\by(\bx)) = p, & \quad \forall \bx\in\R^n.\label{eq:decoupling-full-rank}
    \end{align}
\end{definition}
Given an output with relative degree $r$, define:
\begin{equation}\label{eq:output-coordinates}
      \vec{\mb{y}}(\bx) \coloneqq 
    \begin{bmatrix}
        \by(\bx)^\top \;L_{\bf}\by(\bx)^\top \;  \cdots \;L_{\bf}^{r-1}\by(\bx)^\top
    \end{bmatrix}^\top
    \in\real^{pr},
\end{equation}
as a new set of partial coordinates
with dynamics:
\begin{align}\label{eq: linear output dynamcis}
    \frac{\mathrm{d}}{\dt}\vec{\mb{y}}(\bx) & = 
    \underbrace{
    \begin{bmatrix}
        \mb{0} & \mb{I}_{p(r-1)}\\
        \mb{0} & \mb{0}
    \end{bmatrix}}_{\mb{A}} \vec{\mb{y}}(\bx)
    \! + \!
    \underbrace{
    \begin{bmatrix}
        \mb{0} \\
        \mb{I}_p
    \end{bmatrix}}_{\mb{B}}
    \mb{v}\\
    \mb{v} &\coloneqq L_{\bf}^{r}\by(\bx) + L_{\bg}L_{\bf}^{r-1}\by(\bx)\bu,     \label{eq: y_r mapping}
\end{align}
where \eqref{eq: y_r mapping} is viewed as an input to \eqref{eq: linear output dynamcis}. The output dynamics in \eqref{eq: linear output dynamcis} are a chain of integrators and techniques such as \cite{AndrewCDC22,MurrayACC20} may be employed to construct CBFs. Importantly, when $\by$ has relative degree $r$, any controller $\bv= \hat{\bk}(\vec{\by})$ designed for \eqref{eq: linear output dynamcis} may be transferred back to \eqref{eq:ol-sys} via:
\begin{equation}\label{eq:input-transformation}
    \bu = L_{\bg}L_{\bf}^{r-1}\by(\bx)^\dagger\left[\hat{\bk}(\vec{\by}(\bx)) - L_{\bf}^{r}\by(\bx) \right],
\end{equation}
where the right psuedo-inverse $L_{\bg}L_{\bf}^{r-1}\by(\bx)^\dagger$ exists given \eqref{eq:decoupling-full-rank}. When the output coordinates $\vec{\by}$ are physically relevant to the original safety specification for \eqref{eq:ol-sys}, this provides a systematic approach to synthesizing CBFs and safety-critical controllers for complex nonlinear systems. In general, however, the outputs relevant to the safety specification for \eqref{eq:ol-sys} may not have a valid relative degree, precluding the ability to directly transfer inputs from the output integrator system \eqref{eq: linear output dynamcis} back to the nonlinear system \eqref{eq:ol-sys} via the controller transformation \eqref{eq:input-transformation}. In what follows, we provide a framework for relating inputs of the output integrator system \eqref{eq: linear output dynamcis} to those of the nonlinear system \eqref{eq:ol-sys} under weaker conditions than Def. \ref{def:relative-degree}, and demonstrate how this leads to the synthesis of CBFs for practically relevant systems. 

\input{main_theory_gilbert}

\input{simulations}

\renewcommand{\baselinestretch}{1.08}
\bibliography{alias,PO,GB, main-GB, cosner,cohen}
\bibliographystyle{ieeetr}
\input{appdx}
\end{document}

%% file: def.tex
\newcommand{\setmap}[3]{#1:#2 \mathrel{\vcenter{\mathsurround0pt
\ialign{##\crcr
		\noalign{\nointerlineskip}$\rightarrow$\crcr
		\noalign{\nointerlineskip}$\rightarrow$\crcr
		}}}%
		#3}

\newcommand{\naturals}{\mathbb{N}}
\newcommand{\re}{\mathbb{R}}
\newcommand{\real}{\mathbb{R}}
\newcommand{\realnonneg}{\mathbb{R}_{\ge 0}}
\newcommand{\realpos}{\mathbb{R}_{> 0}}
\newcommand{\until}[1]{[#1]}
\newcommand{\map}[3]{#1:#2 \rightarrow #3}
\newcommand{\qedA}{~\hfill \ensuremath{\square}}
\newcommand\scalemath[2]{\scalebox{#1}{\mbox{\ensuremath{\displaystyle #2}}}}
\newcommand{\interior}{\operatorname{int}}

\newcommand{\longthmtitle}[1]{\mbox{}{\textit{(#1):}}}
\newcommand{\setdef}[2]{\{#1 \; | \; #2\}}
\newcommand{\setdefb}[2]{\big\{#1 \; | \; #2\big\}}
\newcommand{\setdefB}[2]{\Big\{#1 \; | \; #2\Big\}}
\newcommand*{\SetSuchThat}[1][]{} 
\newcommand*{\MvertSets}{%
    \renewcommand*\SetSuchThat[1][]{%
        \mathclose{}%
        \nonscript\;##1\vert\penalty\relpenalty\nonscript\;%
        \mathopen{}%
    }%
}
\MvertSets 
\DeclarePairedDelimiterX \Set [2] {\lbrace}{\rbrace}
    {\,#1\SetSuchThat[\delimsize]#2\,}

\newcommand{\dt}{\mathrm{d}t}
\newcommand{\dy}{\mathrm{d}y}
\newcommand{\dtau}{\mathrm{d}\tau}
\newcommand{\hout}{h_\mathrm{out}}
\newcommand{\dhout}{\dot{h}_\mathrm{out}}
\newcommand{\ThetaD}{\Theta_{\mathrm{D}}}
\newcommand{\vy}{\vec{\mb{y}}}
\newcommand{\ud}{\bu_{1,\mathrm{D}}}

\newcommand{\Cc}{\mathcal{C}}
\newcommand{\Tc}{\mathcal{T}}
\newcommand{\Dc}{\mathcal{D}}
\newcommand{\Hc}{\mathcal{H}}
\newcommand{\Kc}{\mathcal{K}}
\newcommand{\Pc}{\mathcal{P}}
\newcommand{\Uc}{\mathcal{U}}
\newcommand{\Sc}{\mathcal{S}}
\newcommand{\Xc}{\mathcal{X}}
\newcommand{\Yc}{\mathcal{Y}}
\newcommand{\Vc}{\mathcal{V}}
\newcommand{\Zc}{\mathcal{Z}}
\newcommand{\Ec}{\mathcal{E}}
\newcommand{\Rm}{\mathcal{\mathbb{R}}}
\DeclarePairedDelimiter{\Brack}{\lBrack}{\rBrack}
\newcommand{\bbrack}[1]{{
  \mathchoice
    {\left\lbrack \!\left\lbrack #1 \right\rbrack \!\right\rbrack} 
    {\left\lbrack\!\left\lbrack #1 \right\rbrack\!\right\rbrack} 
    {} 
    {} 
  }
}

\newcommand{\defeq}{\triangleq}

\newcommand{\vr}{\varepsilon}
\newcommand{\nom}{{\operatorname{nom}}}
\newcommand{\m}{{\operatorname{min}}}
\newcommand{\des}{{\operatorname{des}}}
\newcommand{\on}{{\operatorname{on}}}
\newcommand{\off}{{\operatorname{off}}}
\newcommand{\fl}{{\operatorname{FL}}}
\newcommand{\Lie}{\mathcal{L}}
\newcommand{\qp}{{\operatorname{QP}}}

\newcommand{\ie}{i.e., }
\newcommand{\todo}[1]{{\color{cyan} Todo: #1}}

\newcommand{\ba}{\mathbf{a}}
\newcommand{\bb}{\mathbf{b}}
\newcommand{\be}{\mathbf{e}}
\renewcommand{\bf}{\mathbf{f}} 
\newcommand{\bc}{\mathbf{c}}
\newcommand{\bff}{\mathbf{f}}
\newcommand{\bg}{\mathbf{g}}
\newcommand{\dhat}{\hat{\mathbf{d}}}
\newcommand{\bk}{\mathbf{k}}
\newcommand{\bp}{\mathbf{p}}
\newcommand{\bq}{\mathbf{q}}
\newcommand{\bu}{\mathbf{u}}
\newcommand{\bv}{\mathbf{v}}
\newcommand{\bx}{\mathbf{x}}
\newcommand{\by}{\mathbf{y}}
\newcommand{\bz}{\mathbf{z}}
\newcommand{\bA}{\mathbf{bA}}
\newcommand{\bB}{\mathbf{B}}
\newcommand{\bD}{\mathbf{D}}
\newcommand{\bC}{\mathbf{C}}
\newcommand{\bF}{\mathbf{F}}
\newcommand{\bJ}{\mathbf{J}}
\newcommand{\bG}{\mathbf{G}}
\newcommand{\bK}{\mathbf{K}}
\newcommand{\bP}{\mathbf{P}}
\newcommand{\bW}{\mathbf{W}}
\newcommand{\bY}{\mathbf{Y}}
\newcommand{\bQ}{\mathbf{Q}}

\newcommand{\mb}[1]{\mathbf{#1}}
\newcommand{\R}{\mathbb{R}}
\newcommand{\bs}[1]{\boldsymbol{#1}}

\newcommand{\btau}{\boldsymbol{\tau}}
\newcommand{\bomega}{\boldsymbol{\omega}}

\newcommand{\bpo}{\bp_{\mathrm{obs}}}

\newcommand{\gilbertSys}{\textcolor{blue}{Gilbert system}}

%% file: main_theory_gilbert.tex
\section{Main Theoretical Contribution}
\subsection{Systems with Dual Relative Degree}\label{sec:new-relative-degree}
In this section, we present a methodology to synthesize CBFs and safety-critical controllers for systems whose outputs may not have a valid relative degree, but satisfy other desirable properties that enable the construction of CBFs. We characterize these systems using the notion of \emph{dual relative degree}, which captures the situation in which inputs influence the outputs at two different orders of differentiation.  
To facilitate our approach, we assume that \eqref{eq:ol-sys} has multiple control inputs, i.e., $m\geq2$, and thus may be written as:
\begin{equation}\label{eq:multi-input-system}
    \dot{\bx} = \bf(\bx) + \underbrace{\bg_{1}(\bx)\bu_{1} + \bg_{2}(\bx)\bu_{2}}_{\bg(\bx)\bu},
\end{equation}
where $\bu_{1}\in\R^{m_1}$, $\bu_{2}\in\R^{m_2}$ such that $m=m_1 + m_2$ with $\bu=(\bu_1,\bu_2)$, while $\bg_1\,:\,\R^n\rightarrow\R^{n\times m_1}$ and $\bg_2\,:\,\R^n\rightarrow\R^{n\times m_2}$ decompose $\bg$ as $\bg(\bx)=\begin{bmatrix}\bg_1(\bx) &\bg_{2}(\bx) \end{bmatrix}$. Given these dynamics and an output $\by\,:\,\R^n\rightarrow\R^p$ for \eqref{eq:multi-input-system}, the inputs affect these outputs via:
\begin{equation}
    \underbrace{
    L_{\bg}L_{\bf}^{i}\by(\bx)}_{p\times m} = \big[\;\underbrace{L_{\bg_1}L_{\bf}^i\by(\bx)}_{p\times m_1}  \; \underbrace{L_{\bg_2}L_{\bf}^i\by(\bx)}_{p\times m_2} \;\big].
\end{equation}
Rather than requiring $\by$ to have a relative degree, we will require it to have a dual relative degree, defined as follows:
\begin{definition}\label{def:dual-relative-degree}(Dual Relative Degree)
    A multi-input system \eqref{eq:multi-input-system} with smooth output $\by\,:\,\R^n\rightarrow\R^p$ is said to have \emph{dual relative degree} $(r,q)\in\mathbb{N}\times\mathbb{N}$ 
    if \eqref{eq:zero-lie-derivatives} holds and for all $\bx\in\R^n$:
    \begin{align}
        L_{\bg_2}L_{\bf}^{r-1}\by(\bx)= \mathbf{0}, \label{eq: Lg_2h = 0 condition}\\
        \mathrm{rank}(L_{\bg_1}L_{\bf}^{r-1}\by(\bx)) = m_1, \label{eq: Lg_1h none 0 condition} \\
        \mathrm{rank}(L_{\bg_2}L_{\bf}^{q-1}L_{\bg_1}L_{\bf}^{r-1}\by(\bx))=m_2.\label{eq:an-insane-number-of-lie-derivatives}
    \end{align}
\end{definition}
Dual relative degree systems represent those whose inputs influence the output at two different levels of differentiation, and capture systems such as unicycles and quadrotors. While we will not explicitly leverage \eqref{eq:an-insane-number-of-lie-derivatives}, it is often implicit in our other assumptions (e.g., on the existence of a tracking control Lyapunov function in Def. \ref{def:tracking-clf}) and is thus included to better characterize the systems to which our approach applies.

Similar to Def. \ref{def:relative-degree}, when 
\eqref{eq:multi-input-system} has a dual relative degree,
we may define a set of output coordinates and corresponding output dynamics as in \eqref{eq:output-coordinates} and \eqref{eq: linear output dynamcis}, respectively. However, when $\by$ does not have a relative degree in the sense of Def. \ref{def:relative-degree}, there does not exist a one-to-one correspondence between inputs of \eqref{eq: linear output dynamcis} and \eqref{eq:multi-input-system}. Despite this, 
if \eqref{eq:multi-input-system} has a dual relative degree according to Def. \ref{def:dual-relative-degree},
then  given a desired controller $\hat{\bk}\,:\,\R^{pr}\rightarrow\R^p$ for the linear output dynamics \eqref{eq: linear output dynamcis}, we may 
find the input $\bu_1$ which actuates the outputs in the manner closest to that of 
$\hat{\mb{k}}$ 
via least-squares minimization:
\begin{align}
    \bk_{1}(\bx) := & \argmin_{\bu_1\in\R^{m_1}}\|L_{\bf}^{r}\by(\bx) + L_{\bg_1}L_{\bf}^{r-1}\by(\bx)\bu_1 - \hat{\bk}(\vec{\by}(\bx))\|^2 \nonumber
 \\
    = & L_{\bg_1}L_{\bf}^{r-1}\by(\bx)^{\dagger}\left[\hat{\bk}(\vec{\by}(\bx)) - L_{\bf}^r\by(\bx)\right], \label{eq:k1}
\end{align}
where $L_{\bg_1}L_{\bf}^{r-1}\by(\bx)^{\dagger}$ is the left pseudo-inverse, which exists under the rank assumption \eqref{eq: Lg_1h none 0 condition} from Def. \ref{def:dual-relative-degree}. Taking $\bu_1=\bk_{1}(\bx)$ produces the partial closed-loop system dynamics:
\begin{align}\label{eq:partial-closed-loop-system}
    \dot{\bx} &= \bf(\bx) + \bg_1(\bx)\bk_1(\bx) + \bg_2(\bx)\bu_2 \\
    &\eqqcolon \bf_1(\bx)+ \bg_2(\bx)\bu_2. \label{eq:partial-closed-loop-system2}
\end{align}
While $\bk_1(\mb{x})$ produces inputs closest to $\hat{\mb{k}}(\vec{\mb{y}}(\mb{x}))$, it may not be able to completely eliminate the error between the output actuation $\mb{v}$ in \eqref{eq: y_r mapping} and the desired linear actuation $\hat{\mb{k}}(\vec{\by}(\mb{x}))$. We write this error explicitly as:
\begin{align}\label{eq:tracking-error}
    & \mb{e}(\bx) \coloneqq L_{\bf}^{r}\by(\bx) + L_{\bg_1}L_{\bf}^{r-1}\by(\bx)\bk_1(\bx) - \hat{\bk}(\vec{\by}(\bx)),
    \\
    & = \left(L_{\bg_1}L_{\bf}^{r-1}\by(\bx)L_{\bg_1}L_{\bf}^{r-1}\by(\bx)^\dagger - \mathbf{I}  \right)( \hat{\mb{k}}(\vec{\mb{y}}(\mb{x}))  - L_\mb{f}^r\mb{y}(\mb{x}) ) \nonumber   
\end{align}
which must be compensated for to ensure safety. This will be achieved using Lyapunov-based techniques.

\begin{definition}[Tracking Control Lyapunov Function]\label{def:tracking-clf}
    A continuously differentiable function $V: \R^n \to \R_{\geq 0}$ is a \textit{tracking control Lyapunov function} (CLF) for a control affine system \eqref{eq:ol-sys} with respect to  error function $\mb{e}: \R^n \to \R^{m_1}$ 
    if there exists $\beta, \lambda > 0 $ such that for all $\mb{x} \in \R^n $: 
    \begin{align}
        V(\mb{x})  \geq & \beta \Vert \mb{e}(\mb{x}) \Vert^2 \textup{ and } \label{eq:clf-lower-bound}\\ 
        \inf_{\mb{u} \in \R^m } L_{\mb{f}}V(\mb{x}) + & L_{\mb{g}}V(\mb{x}) \mb{u} \leq -\lambda V(\mb{x}). \label{eq:k2}
    \end{align}
\end{definition}
We will use this tracking CLF to ensure convergence of our error $\mb{e}(\mb{x})$ to zero for the partial closed-loop system \eqref{eq:partial-closed-loop-system2}.

\subsection{CBF Synthesis for Dual Relative Degree Systems}
We now demonstrate how the paradigm in Sec. \ref{sec:new-relative-degree} may be used to synthesize safety-critical controllers. For this, we consider a dual relative degree system of the form \eqref{eq:multi-input-system} with an output $\by\,:\,\R^n\rightarrow\R^p$ 
and output dynamics (\ref{eq: linear output dynamcis},\ref{eq: y_r mapping}). 
We then consider a desired safe set on the output coordinates $\vec{\mb{y}}$:
\begin{equation}\label{eq:C-y}
    \mathcal{C}_{\by}\coloneqq\{\bx\in\R^n \,| \, h_{0}(\vec{\by}(\bx)) \geq 0 \},
\end{equation}
and suppose that $h_0\,:\,\R^{pr}\rightarrow\R$ is a CBF for \eqref{eq: linear output dynamcis} 
with $\bv$ viewed as a ``virtual" input to the linear system.
This assumption guarantees the existence of a smooth\footnote{As discussed in \cite{CohenLCSS23,cohen2024constructive,CohenARC24}, the existence of CBF (or ISSf-CBF) satisfying \eqref{eq:cbf_constraint} (or  \eqref{eq:issf}) with a \emph{strict} inequality guarantees the existence of a controller, as smooth as the dynamics and CBF, satisfying the corresponding barrier condition. Thus, if $h_0$ is a CBF for \eqref{eq: linear output dynamcis}, we may, without loss of generality, construct a smooth feedback controller $\hat{\bk}$ satisfying \eqref{eq:issf-cbf-linear}, with examples of such controllers available in \cite{CohenLCSS23,cohen2024constructive,CohenARC24}.} controller $\hat{\bk}\,:\,\R^{pr}\rightarrow\R^p$ enforcing the ISSf-CBF condition \cite{CohenLCSS23}:
\begin{equation}\label{eq:issf-cbf-linear}
    \pdv{h_0}{\vec{\by}}(\vec{\by})\left[\mb{A}\vec{\by} + \mb{B}\hat{\bk}(\vec{\by}) \right] > - \gamma h_0(\vec{\by}) + \frac{1}{\varepsilon}\left\Vert \pdv{h_0}{\vec{\by}}(\vec{\by})\mb{B} \right\Vert^2,
\end{equation}
for all $\vec{\by}\in\R^{pr}$ for some $\gamma,\varepsilon>0$. While $h_0$ is a CBF for \eqref{eq: linear output dynamcis} with relative degree $r$, it is not necessarily a CBF for \eqref{eq:multi-input-system} with dual relative degree $(r,q)$, and it may be impossible to apply $\hat{\bk}$ to \eqref{eq:multi-input-system} directly.

Inspired by the reduced-order model methods of \cite{TamasRAL22}, we synthesize a CBF for the system with dual relative degree \eqref{eq:multi-input-system} by augmenting $h_0$ with a scaled tracking CLF, $\frac{-1}{\mu}V(\mb{x})$ for some $\mu > 0 $, to account for the error $\mb{e}(\mb{x})$ between $\mb{k}_1 $ and $\hat{\mb{k}}$. 
We formally define this construction as: 
\begin{definition}[Dual Relative Degree CBF (DRD-CBF)]
    Consider system \eqref{eq:multi-input-system} with dual relative degree $(r,q)$.
    If $h_0: \R^{pr} \to \R$ is a CBF for the linear system \eqref{eq: linear output dynamcis} with degree $r$, $\hat{\mb{k}}: \R^{pr} \to \R^p$ is a continuously differentiable function satisfying \eqref{eq:issf-cbf-linear} for some $\gamma, \varepsilon > 0$, and  $V: \R^n \to \R_{\geq 0 }$ is a tracking control Lyapunov function for \eqref{eq:partial-closed-loop-system} with respect to error function \eqref{eq:tracking-error} for some $\beta,\lambda > 0 $, then the function:
    \begin{gather}\label{eq:our-cbf}
        h(\bx) \coloneqq h_0(\vec{\by}(\bx)) - \frac{1}{\mu}V(\bx) \\
        \textup{ with } \mu > 0 \textup{ such that } \lambda \geq \gamma + \frac{\varepsilon \mu }{4 \beta }, \label{eq:parameter_assumption}
    \end{gather}
    is a \textit{dual relative degree CBF} (DRD-CBF) for \eqref{eq:multi-input-system}.
\end{definition}

The condition in \eqref{eq:parameter_assumption} dictates the relation between the convergence rate $\lambda$ of the tracking CLF, 
the safety of $h_0$ (determined by $\hat{\mb{k}}$) via $\gamma$ and the ISSf constant $\varepsilon$, and the scaling parameters $\mu$ and $\beta$. Intuitively, the condition \eqref{eq:parameter_assumption} can be satisfied by increasing the error tracking speed of $V$ by increasing $\lambda$, increasing the conservatism of $\hat{\mb{k}}$ by decreasing $\gamma$ and $\varepsilon$, or by balancing the scaling of $h_0$ and $V$ via $\mu$ or balancing the scaling of $V$ and $\mb{e}$ via $\beta$.  

Next, in Theorem \ref{thm:main}, we prove that all DRD-CBFs are valid CBFs for system \eqref{eq:multi-input-system} by showing that the existence of control actions, derived from $\hat{\mb{k}}$ and the tracking CLF, certifies that \eqref{eq:our-cbf} satisfies the CBF constraint \eqref{eq:cbf_constraint}. Thus, we show that DRD-CBFs are a special class of CBFs for dual relative degree systems that can be directly synthesized using a CBF $h_0$ for a linear integrator system \eqref{eq: linear output dynamcis} and a tracking CLF $V$. 

\begin{theorem}\label{thm:main}
    Consider a system of the form \eqref{eq:multi-input-system} with dual relative degree ($r, q$). If $h: \R^n \to \R$ is a dual relative degree CBF for \eqref{eq:multi-input-system} as in \eqref{eq:our-cbf}, then it is also a CBF \textcolor{black}{and any Lipschitz controller satisfying \eqref{eq:cbf_constraint} for $h$ renders $\mathcal{C}= \{ \mb{x}\in \R^n ~|~ h(\mb{x}) \geq 0 \} \subset \mathcal{C}_{\mb{y}}$ safe}.     
\end{theorem}
\begin{proof}
    Computing the time-derivative of $h_0$ and bounding (omitting dependencies on $\bx$ for brevity) we obtain:
    \begin{align}
        \dot{h}_0 = & \pdv{h_0}{\vec{\by}}(\vec{\by})\left[\mb{A}\vec{\by} + \mb{B}\bv \right] \label{eq:ref output linear dynamcis}\\
        = & \pdv{h_0}{\vec{\by}}(\vec{\by})\left[\mb{A}\vec{\by} + \mb{B}\hat{\bk}(\vec{\by}) \right] + \pdv{h_0}{\vec{\by}}(\vec{\by})\mb{B}\left[\bv - \hat{\bk}(\vec{\by}) \right] \nonumber \\
        > & -\gamma h_0(\vec{\by}) + \frac{1}{\varepsilon}\left\Vert \pdv{h_0}{\vec{\by}}(\vec{\by})\mb{B} \right\Vert^2 \label{eq: Issf liine} \\
        &- \left\Vert \pdv{h_0}{\vec{\by}}(\vec{\by})\mb{B} \right\Vert \left\Vert \bv - \hat{\bk}(\vec{\by})\right\Vert \nonumber \\
        \geq & -\gamma h_0(\vec{\by}) - \frac{\varepsilon}{4}\left\Vert \bv - \hat{\bk}(\vec{\by})\right\Vert^2 \label{eq: complete square}\\
        = & -\gamma h_0(\vec{\by}) - \frac{\varepsilon}{4}\left\Vert L_{\bf}^{r}\by + L_{\bg_1}L_{\bf}^{r-1}\by\bk_1 - \hat{\bk}(\vec{\by})\right\Vert^2 \label{eq: using v definition}\\
        \geq & -\gamma h_0(\vec{\by}) - \frac{\varepsilon}{4\beta}V 
        =  -\gamma h - \left(\frac{\gamma}{\mu} + \frac{\varepsilon}{4\beta}  \right)V. \label{eq; final h0 bound}
    \end{align}
    In the above expression, \eqref{eq:ref output linear dynamcis} follows directly from the linear output dynamics \eqref{eq: linear output dynamcis}. Next, \eqref{eq: Issf liine} is obtained by adding zero, assuming $\hat \bk$ enforces the ISSf-CBF inequality \eqref{eq:issf-cbf-linear} for \eqref{eq: linear output dynamcis}, and applying the Cauchy-Schwartz inequality\footnote{Given $|\ba^\top \bb| \leq \|\ba\|\|\bb\|$ for all $\ba,\bb \in \real^n,\| \cdot\| \coloneq \|\cdot\|_2.$ Setting $\bb \!=\! -\bc$ gives  $-\ba^\top \bc \leq |\!-\ba^\top \bc| \!\leq \!\|\ba\|\|\!-\!\bc\| \!=\! \|\ba\|\|\bc\| \! \! \implies \! \!\ba^\top \bc \geq -\|\ba\|\|\bc\|$.}. We then complete the square to achieve \eqref{eq: complete square}, and then use the definition of $\bv$ in \eqref{eq: y_r mapping} to rewrite \eqref{eq: complete square} as \eqref{eq: using v definition}. Next, we select $\bk_1$ provided in \eqref{eq:k1} and use \eqref{eq:clf-lower-bound} to bound \eqref{eq: using v definition} 
    using $V$. Finally, we use \eqref{eq:our-cbf} to express $h_0$ in terms of $h$ and $V$ to yield \eqref{eq; final h0 bound}.
    
    Since $V$ is a tracking CLF for \eqref{eq:partial-closed-loop-system}, then, for each $\bx\in\R^n$ there exists an input $\bu_2\in\R^{m_2}$ satisfying:
    \begin{equation}\label{eq:tracking-clf-proof}
        L_{\bf_1}V(\bx) + L_{\bg_2}V(\bx)\bu_2 \leq -\lambda V(\bx).
    \end{equation}
    Now, computing the time derivative of $h$ with $\bu_1=\bk_1(\bx)$ from \eqref{eq:k1} and bounding at each $\bx\in\R^n$ using the above expression, we obtain:
    \begin{align}
        \dot{h} &=  \dot{h}_0 - \frac{1}{\mu}\dot{V} =  \dot{h}_0 - \frac{1}{\mu}\pdv{V}{\bx}\left[\bf + \bg_1\bk_1 + \bg_2\bu_2 \right] \label{eq: CLF with k1 and k2}\\
        &=  \dot{h}_0 - \frac{1}{\mu}\pdv{V}{\bx}\left[\bf_1 + \bg_2\bu_2 \right] 
        \geq  \dot{h}_0 + \frac{\lambda}{\mu}V \label{eq: full closed loop bound}\\
       & >  -\gamma h + \frac{1}{\mu}\left(\lambda - \gamma - \frac{\varepsilon\mu}{4\beta} \right)V \label{eq: subs dot h0} 
        \geq  -\gamma h, 
    \end{align}
    where we used the partial closed-loop dynamics \eqref{eq:partial-closed-loop-system} to rewrite $\dot{h}$ in  \eqref{eq: CLF with k1 and k2}. 
    We then select $\bu_2$ that satisfies \eqref{eq:tracking-clf-proof} to obtain the bound in \eqref{eq: full closed loop bound}.
    We then substitute the bound obtained in \eqref{eq; final h0 bound} for $\dot h_0$ to obtain the first bound in \eqref{eq: subs dot h0}. Finally, applying the inequality \eqref{eq:parameter_assumption} for $\lambda$ yields the second bound in \eqref{eq: subs dot h0}. Given that this choice of $\bu$ guarantees $\dot h (\bx, \bu) > -\gamma(h(\bx))$ for all $\bx \in \real^n$, $h$ is a valid CBF\footnote{Implicit in the fact that $h$ satisfies \eqref{eq:cbf_constraint} with a strict inequality is that $\pdv{h}{\bx}(\bx)\neq\bm{0}$ for all $\bx\in\partial\mathcal{C}$, a regularity condition needed to apply standard CBF results regarding forward invariance \cite{AmesTAC17}.} for \eqref{eq:multi-input-system}. \textcolor{black}{Furthermore, since $h$ is a CBF for \eqref{eq:multi-input-system}, any Lipschitz continuous controller that satisfies \eqref{eq:cbf_constraint} renders $\mathcal{C}$ safe \cite[Cor. 2]{AmesTAC17}, and since $V(\mb{x}) \geq 0$, the safe set $\mathcal{C}$ is contained in the desired safe set $\mathcal{C}_{\by}$, $\mathcal{C}\subset \mathcal{C}_{\by}$, so trajectories that are safe with respect to $\mathcal{C}$ also remain in the desired safe set $\mathcal{C}_\by$. }
\end{proof}

\begin{remark}
    Theorem \ref{thm:main} and its proof are inspired by results on CBFs and reduced-order models \cite{TamasRAL22,CohenARC24} in which CBFs for simple models are combined with Lyapunov functions certifying tracking of such models to establish safety of the overall system. Our CBF construction may also be seen through this lens: CBFs designed for a chain of integrators \eqref{eq: linear output dynamcis} are transferred to a high-dimensional nonlinear system with dual relative degree using Lyapunov-based techniques. Importantly, as demonstrated in the following sections, the conditions of Theorem \ref{thm:main} are shown to hold for relevant systems, such as quadrotors, allowing for systematic synthesis of safety-critical controllers for highly dynamic systems based on CBFs designed for a chain of integrators.
\end{remark}

The preceding result requires the existence of a global CLF, that is, Def \ref{def:tracking-clf} is presented for all $\bx \in \re^n$. Due to various factors (e.g., topological obstructions to continuous stabilization \cite[Ch. 4]{LiberzonSwitching}), such a CLF may not exist for a given system of interest (e.g., that with states evolving on a differentiable manifold), and \eqref{eq:k2} may only hold on a set $\mathcal{D}\subset\R^n$. We relax the global condition on \eqref{eq:k2}, and show below that, while global stabilization in such a situation may not be possible, enforcing safety globally is still possible.

%

\begin{corollary}(Global Safety)\label{corollary: main}
Let the conditions of Theorem \ref{thm:main} hold, but suppose that \eqref{eq:k2} only holds on a set $\mathcal{D}\subset\R^n$. Define $\mathcal{E}\coloneqq\mathbb{R}^n\setminus\mathcal{D}$. Provided that for all $\bx\in\mathcal{E}$:
\begin{equation}\label{eq: corollary condition}
\begin{aligned}
    \left[L_{\bg_1}h_0(\vec{\by}(\bx)) = \frac{1}{\mu}L_{\bg_1}V(\bx) \right] \\
    \implies \left[L_{\bf}h_0(\vec{\by}(\bx)) - \frac{1}{\mu}L_{\bf}V(\bx) \geq -\gamma h(\bx) \right],
\end{aligned}
\end{equation}
then $h$ is a CBF for \eqref{eq:multi-input-system}.

\end{corollary}

\begin{proof} We divide the proof by considering two cases:
\\
\noindent
\textbf{Case 1:} If $\bx \in\Dc$, then
\eqref{eq:k2} holds and Theorem \ref{thm:main} implies that \eqref{eq:cbf_constraint} holds for all $\bx\in\Dc$.
\\
\noindent
\textbf{Case 2:} If $\bx \in \mathcal{E}$, then \eqref{eq:k2} does not hold and we must have $L_{\bg_2}V(\bx)=\mathbf{0}$ (\ie if $L_{\bg_2}V(\bx)\neq\mathbf{0}$ at this point, then there always exists an input $\bu_2$ satisfying \eqref{eq:k2}).
Taking the time derivative of $h$ at $\bx\in\mathcal{E}$ yields:
  \begin{align}
        \dot{h} &= L_\bf h_0 - \frac{1}{\mu} L_\bf V +  \left(L_{\bg_1} h_0 - \frac{1}{\mu}L_{\bg_1} V\right)\bu_1 - \frac{1}{\mu}L_{\bg_2} V\bu_2 \nonumber \\
        &= \underbrace{L_\bf h_0 - \frac{1}{\mu} L_\bf V}_{L_{\bf}h(\bx)} +  \underbrace{\left(L_{\bg_1} h_0 - \frac{1}{\mu}L_{\bg_1} V\right)}_{L_{\bg_1}h(\bx)}\bu_1,
\end{align}
for all $\bx\in\mathcal{E}$,
where $\bg(\bx)=\begin{bmatrix}
    \bg_1(\bx) & \bg_2(\bx)
\end{bmatrix}$. Provided that \eqref{eq: corollary condition} holds, we have
$L_{\bg}h(\bx)=\mathbf{0}$ implies that $L_{\bf}h(\bx)\geq - \gamma h(\bx)$ for all $\bx\in\mathcal{E}$, which implies that \eqref{eq:cbf_constraint} holds for all $\bx\in\mathcal{E}$ (cf. \cite{jankovic2018robust}). Combining Cases 1 and 2, we have that \eqref{eq:cbf_constraint} holds for all $\bx\in\R^n$, implying that $h$ is a CBF for \eqref{eq:multi-input-system}.
\end{proof}
Given $h$ in \eqref{eq:our-cbf}, Theorem \ref{thm:main} and the above Corollary allow for synthesizing an optimization-based controller as in \eqref{eq:cbfqp} for any given $\gamma \in \real_{>0}$ and nominal controller $\bk_\nom$.
The following sections provide case studies of systems that satisfy the conditions of Theorem \ref{thm:main} and Corollary
\ref{corollary: main}.

\section{Case Study: Unicycle with Drift}\label{sec:unicycle}
We begin by considering the unicycle system with drift:
    \begin{equation}\label{eq: unicycle system w drift}
        \frac{\mathrm{d}}{\dt}\begin{bmatrix}
            x \\ y \\ \theta
        \end{bmatrix}
        =
  \underbrace{\begin{bmatrix}
            d_x \\ d_y \\ 0
        \end{bmatrix}}_{{\bf}(\bx)} 
        + \underbrace{\begin{bmatrix}
            \cos(\theta) \\ \sin(\theta) \\ 0
        \end{bmatrix}}_{\bg_1(\bx)}
        v
        +
        \underbrace{\begin{bmatrix}
            0 \\ 0 \\ 1
     \end{bmatrix}}_{\bg_2(\bx)}\omega,
    \end{equation}
where the state $\bx = (x, y, \theta) \in \Xc = \real^2 \times \mathbb{S}^1$ defines the planar position and heading angle. The control input $\bu= (\bu_1, \bu_2) = (v,\omega)  \in \real^2$ represents the linear and angular velocities. The values $d_x, d_y \in \real$ represent constant drift, motivated by the unicycle operating on a treadmill (e.g. Fig. \ref{fig:unicycle ellipse}). Our control objective is to constrain the position of the unicycle. Thus, we take our outputs as $\by(\bx) = (x, y) \in \real^2$, which do not have a valid relative degree in the sense of Def. \ref{def:relative-degree}. However, the unicycle with this choice of outputs has dual relative degree $(r,q)=(1,1)$ as one may verify that $L_{\bg_1}\by(\bx) = [\cos(\theta)\;  \sin(\theta)]^\top$ and $L_{\bg_2}L_{\bg_1}\by(\bx)=[-\sin(\theta)\;\cos(\theta)]^\top$, which both have rank 1.

\subsection{Safety Specification}
    We consider a safety requirement that ensures the unicycle remains within an ellipse centered at $\by_\mathrm{c} = [ x_c, y_c]^\top \in \real^2$:
    \begin{align}\label{eq: ellipse CBF}
        h_0(\vec{\by}(\bx)) = 1 - (\by(\bx) - \by_\mathrm{c})^\top P(\by(\bx) - \by_\mathrm{c}),
    \end{align}
    where $P  = \mathrm{diag}(p_1, p_2) \in  \real^{2\times 2}$ is a diagonal matrix and $p_1, p_2 \in \real_{>0}$ are the weights corresponding to the lengths of the major and minor axes of the ellipse. The output coordinates $\vec{\by}(\bx) = \by(\bx)$ yield a single-integrator system of the form \eqref{eq: linear output dynamcis}.
    We design a differentiable controller  $\hat{\bk} \coloneqq [\hat{\bk}_x, \hat{\bk}_y]^\top:\real^2 \rightarrow \real^2$
    satisfying \eqref{eq:issf-cbf-linear} for the single integrator using the methods in \cite{CohenLCSS23}. 
    We then leverage the single integrator controller $\hat{\bk}$ to generate a safe linear velocity $v = \bk_1(\bx)$ as in \eqref{eq:k1} for the unicycle. 

\subsection{Safety-Critical Control}
Let $\tilde \bk (\bx)  \coloneqq  \hat \bk (\vec{\by})- L_{\bf}\by = [\hat \bk_x(\vec{\by}) - d_x, \hat \bk_y(\vec{\by}) - d_y]^\top$. We consider the tracking CLF, slightly modified from \cite{lee_geometric_2010}:
\begin{align} \label{eq: our-tracking clf explicit}
    V(\bx) &= \frac{\|\tilde \bk (\bx) \|^2}{2}\mathrm{tr}\left( \mathbf{I}_{2 \times 2} - \mathbf{R}(\theta_\des(\bx))^\top \mathbf{R}(\theta)\right),
\end{align}
  where for $\|\tilde \bk (\bx)\| \neq 0$, the direction of the vector $\tilde \bk(\bx)$ provides the desired safe heading angle (\ie safe yaw) as $\theta_\des(\bx) = \mathrm{atan2}(\hat \bk_y(\vec{\by}(\bx)) - d_y, \hat \bk_x(\vec{\by}(\bx)) - d_x)$,  while if $\|\tilde \bk (\bx)\| = 0 $, then $V(\bx) = 0$, making $\theta_\des(\bx)$ a free parameter that may be assigned arbitrarily.
  The term $\mathbf{R} \in SO(2)$ is a 2D rotation matrix, so it follows that  \eqref{eq: our-tracking clf explicit} yields:
  \begin{align}
      V(\bx) = \|\tilde \bk (\bx) \|^2 (1 - \cos(\theta - \theta_\des(\bx))),
  \end{align} which satisfies \eqref{eq:clf-lower-bound} 
  as shown in the appendix.
%
%
%
%

We now consider the DRD-CBF as in \eqref{eq:our-cbf}, which we  show is a CBF for \eqref{eq: unicycle system w drift}. We first find that $L_{\mb{g}_2}V(\mb{x}) = 0 \implies \theta \in \{ \theta_\textup{des}(\mb{x}), \theta_\textup{des}(\mb{x}) + \pi \} $. 
 Let $\mu = 0.06$ and $\hat \bk(\vec \by(\bx)) = -\rho P^{\frac{1}{2}}\by(\bx)$ with $\rho = 0.16$, then 
 $L_{\mb{g}_1}h(\mb{x}) \neq 0$ when $\theta - \theta_\textup{des}(\mb{x}) = \pi$ for all $\bx \in \Cc_\by$ defined in \eqref{eq:C-y}. Thus, Collorary \ref{corollary: main} applies. 
Note that this does not imply global stability of $\theta$ on $\mathbb{S}^1$ with a continuous controller, but that there exists inputs for each $\bx \in \Cc_\by$ satisfying the CBF condition \eqref{eq:cbf_constraint}, ensuring safety but not necessarily stability of $\theta=\theta_{\textup{des}}(\bx)$.

\begin{figure}
    \centering  \includegraphics[width=1.0\linewidth]{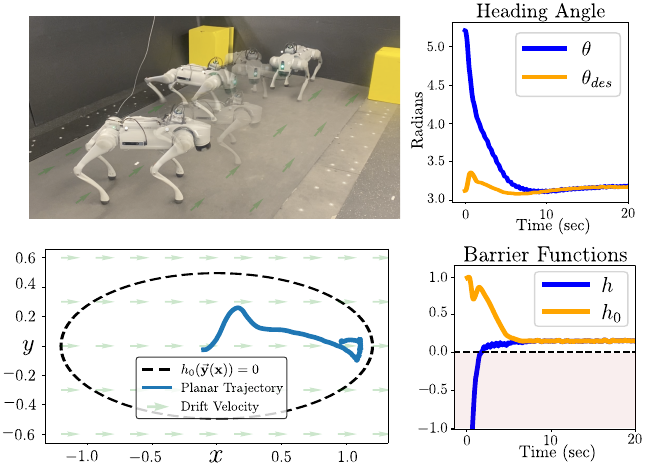}
    \vspace{-4mm}
    \caption{\textbf{(top left)} The quadrupedal robot 
    \textbf{(top right)} The yaw, $\theta$, of the quadruped in blue and the desired yaw, $\theta_\textup{des}$, from the desired safe controller for the linear system $\hat{\mb{k}}(\vec{\mb{y}})$ in orange. 
    \textbf{(bottom left)} ($x, y$) trajectories of the robot in blue with the drift velocity shown using green arrows and the boundary of $\Cc_\by$ \eqref{eq:C-y} shown as a black dotted line. Notably, the trajectories stay in this set and satisfy our safety criterion as desired.  
    \textbf{(bottom right)} Our DRD-CBF $h$ in blue and the safety criterion $h_0$ in orange. Notably, $h_0$ remains above zero.
    The robot is initialized with an unsafe yaw $\theta$, causing $h$ to be initially negative (\ie outside the safe set $\Cc$ \eqref{eq:safe-set}). We demonstrate that the geometric tracking CLF \eqref{eq: our-tracking clf explicit} incorporated in $h$ leads to the convergence of $\theta$ to a safe yaw $\theta_\textup{des}$ yielding a positive $h$, enforcing attraction to $\Cc$.}.
\label{fig:unicycle ellipse}
    \vspace{-6mm}
\end{figure}
\subsection{Unicycle: Simulation and Hardware}
\label{subsec:unicycle_experiments}%
We demonstrate the effectiveness of our proposed CBF \eqref{eq:our-cbf} in ensuring safety for system \eqref{eq: unicycle system w drift} in simulation and hardware. Using the safety specification \eqref{eq: ellipse CBF}, we synthesize a safe controller as in \eqref{eq:cbfqp} with the $h$ defined in \eqref{eq:our-cbf} for \eqref{eq: unicycle system w drift} with drift terms $d_x = 0.35 \ m/s$ and $d_y= 0$.
For the hardware demonstration, we apply this controller to a Unitree GO2 quadruped for which the unicycle may serve as a ROM\footnote{The velocity commands generated by our controller are then tracked by Unitree's onboard velocity tracking controller. In general, such an approach will lead to ISSf of the closed-loop system as analyzed in \cite{TamasRAL22,CohenARC24}.}. On hardware, the drift terms are captured by placing the quadruped on a treadmill moving at a constant velocity of $0.35\,m/s$.
The simulated and real-world trajectories can be seen ensuring safety in Fig.\, \ref{fig:unicycle ellipse} with a nominal controller of zero linear and angular velocity.

To further illustrate our approach, we design a safe controller for \eqref{eq: unicycle system w drift} with no drift using $h_0 = h_\textup{obs}$ as in
\eqref{eq: planar drone safety specification} 
to avoid an obstacle with the nominal tracking controller $k_\nom(\bx) = [K_\mathrm{p}\|\by(\bx) - \by_\des\|, -K_\mathrm{q} \sin(\theta - \theta_\des(\bx))]^\top $ where $\by_\des \in \real^2$ is a goal position and $K_\mathrm{p}, K_\mathrm{q} \in \real_{>0}$ are gains. This $h_0$ is used to construct a DRD-CBF $h$ as in the previous subsection with $\hat{\bk}$
    satisfying \eqref{eq:issf-cbf-linear}, which satisfies the conditions of Theorem \ref{thm:main} and Corollary \ref{corollary: main}. Applying \eqref{eq:cbfqp} to this system using the corresponding $h$ produces the results in Fig.\,\ref{fig:quadruped}, which shows this safe trajectory alongside the unfiltered $k_\nom$ trajectory, which violates safety. 
\begin{figure}
    \centering
    \vspace{-4mm}
    \includegraphics[width=0.9\linewidth]{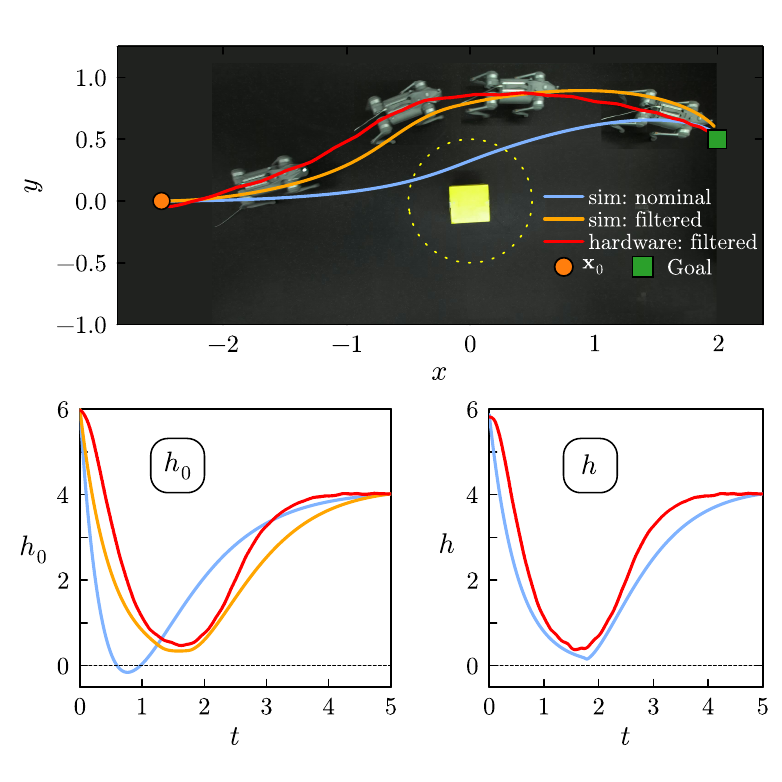}
    \vspace{-6mm}
    \caption{Hardware (red) and simulated (orange and blue) unicycle trajectories. \textbf{(top)} A composite image showing the position of the quadruped throughout the experiment alongside position trajectories for the obstacle avoidance safe controller designed in Section \ref{subsec:unicycle_experiments} on hardware (red) and in simulation, with the unfiltered $\mb{k}_\mathrm{nom}$. \textbf{(bottom left)} The value of output safety requirement $h_0$ throughout the experiment. \textbf{(bottom right)} The value of the DRD-CBF $h$ throughout the experiments.}
    \label{fig:quadruped}
    \vspace{-5mm}
\end{figure}
\section{Case Study: Quadrotor}\label{sec:quadrotor}
We now consider the quadrotor system, as discussed in \cite{lee_geometric_2010}, with the dynamics in the form  \eqref{eq:multi-input-system} given by: 
\begin{align}
      \! \! \! \frac{\mathrm{d}}{\dt}\! 
    \begin{bmatrix}
        \mb{y}\\
        \dot{\mb{y}}\\
        \mb{q}\\ 
        \bs{\omega}
    \end{bmatrix} 
    \! \! = \! \! 
    \underbrace{\begin{bmatrix}
        \dot{\mb{y}}\\
        - g \mb{e}_z\\
        \frac{1}{2}\mb{q} \otimes \bs{\omega}_\mb{q} \\
        \mb{0}
    \end{bmatrix}}_{\bf(\bx)}
    \!\! + \!\!
    \underbrace{\begin{bmatrix}
        \mb{0} \\ 
        \frac{1}{m}\mb{R}(\mb{q}) \mb{e}_z\\
        \mb{0} \\
        - J^{-1}\bs{\omega} \times J \bs{\omega} 
    \end{bmatrix}}_{\bg_1(\bx)} \! \tau 
    + \!
    \underbrace{\begin{bmatrix}
        \mb{0}\\
        \mb{0}\\
        \mb{0}\\
        J^{-1}
    \end{bmatrix}}_{\bg_2(\bx)} \! \!\mb{M} \label{eq:3D_drone}
\end{align}
where $\mb{y} = (x, y,z) \in \R^3$ is the three dimensional position of the center of mass, $\mb{q} \in \mathbb{S}^3$ is the quaternion representing the orientation, $\bs{\omega}$ is the angular velocity in the body-fixed frame, and 
$\bs{\omega}_\mb{q} = (0, \bs{\omega})$ is the pure quaternion representation of $\bs{\omega}$, $\mb{J} \in \R^{3 \times 3} $ is the inertia matrix with respect to the body-fixed frame, $m\in \R_{>0}$ is the total mass, $ \tau \in \R$ is the total thrust, $\mathbf{M} \in \R^3$ is the total moment in the body-fixed frame, $g \in \R_{>0}$ is gravity, and $\mb{e}_z$ is the inertial frame $z$-direction. Additionally, $\mb{R} \in SO(3)$ is the rotation matrix which can be derived from quaternion $\bq$.
%
The outputs $\by=(x,y,z)$ 
do not have a relative degree. Similar to the previous example, this system has dual relative degree $(r,q)=(2,2)$ as both $L_{\bg_1}L_\bf \by(\bx) = \mathbf{R}(\bq)\mb{e}_z$ and $L_{\bg_2}L_{\bf}L_{\bg_1}L_{\bf}\by(\bx)$ have full rank. 
Fixing the $y$ position, and assuming  unit inertial values, \eqref{eq:3D_drone} reduces to the planar model of the quadrotor, where:
    \begin{subequations}\label{eq: planar drone ex}
      \begin{align}
        \bf(\bx) &=
        \begin{bmatrix}
            \dot{x}, 
            \dot{z}, 
            \omega,  
            0,  
            -g,  
            0
        \end{bmatrix}^\top
        \\
        \bg_{1}(\bx)& =
        \begin{bmatrix}
            0 , 
            0,
            0,
            -\sin(\theta),
            \cos(\theta),
            0 
        \end{bmatrix}^\top
        \\
        \bg_{2}(\bx) &=
        \begin{bmatrix}
             0,
             0,
             0,
             0,
             0, 
             1
        \end{bmatrix}^\top
        \end{align}
    \end{subequations}
with the state $\bx=(x,z,\theta,\dot{x},\dot{z},\omega) \in \Xc = \real^2 \times \mathbb{S}^1 \times \real^3$ consisting of the quadrotor's horizontal position $x$, vertical position $z$, its orientation with respect to the horizontal plane $\theta$, and their corresponding rates of change. The input to the system is $\bu= (\bu_1, \bu_2) =(\tau,\mathbf{M}) \in \real^2$ where $\tau$ denotes the total thrust and $\mathbf{M}$ denotes the total moment. Similar to the previous case study, our goal is to constrain the planar position of the quadrotor, captured by the outputs $\by(\bx) = (x,z) \in \real^2$, which leads to a system with dual relative degree $(r,q)=(2,2)$, where $L_{\bg_{1}}L_{\bf}\by(\bx) = [-\sin(\theta), \cos(\theta)]^\top$ and $L_{\bg_2}L_{\bf}L_{\bg_{1}}L_{\bf}\by(\bx) = [\cos(\theta), \sin(\theta)]^\top$.
    \subsection{Planar Quadrotor: Safety Specification}
    To avoid an obstacle at a position $\by_\mathrm{obs} = [x_\mathrm{obs}, y_\mathrm{obs}] \in \real^2$ with radius $r_\mathrm{o} \in \real_{> 0}$, we consider the safety constraint:
    \begin{align}\label{eq: planar drone safety specification}
    h_\mathrm{obs}(\by(\bx)) = \|\by(\bx) - \by_\mathrm{obs} \|^2 - r_\mathrm{o}^2.
    \end{align}
    As this system has dual relative degree $(r,q)=(2,2)$, the output coordinates are given by $\vec{\by}(\bx) = (\by(\bx), \dot \by(\bx))=(x,z,\dot{x},\dot{z})$ and yield double integrator dynamics of the form \eqref{eq: linear output dynamcis} for which techniques such as \cite{AndrewCDC22,MurrayACC20,WeiTAC22} may be employed to construct a CBF from $h_\textup{obs}$. 
   We construct a CBF $h_0: \real^4 \rightarrow \real$ for a double-integrator using \cite{AndrewCDC22}
    and design a controller  $\hat \bk \coloneqq [\hat{\bk}_x, \hat{\bk}_y]^\top:\real^4\rightarrow \real^2$ satisfying \eqref{eq:issf-cbf-linear}
    that generates a safe thrust $\tau =  \bk_1(\bx)$ as in \eqref{eq:k1} for the quadrotor.

\begin{figure}[!t]
\centering
\includegraphics[width=.9\linewidth]{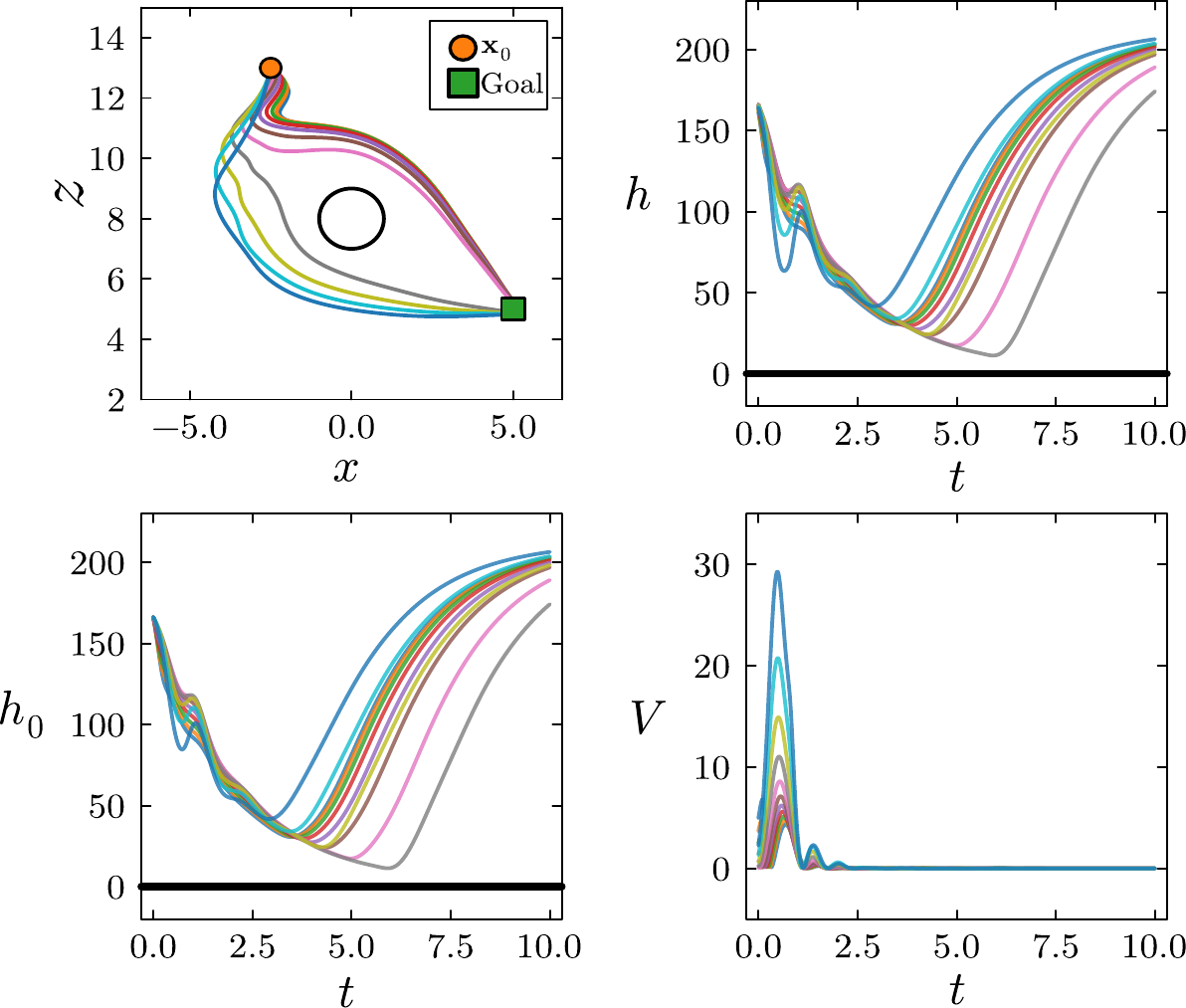}
\vspace{-3mm}
\caption{2D planar quadrotor simulation. \textbf{(top left)} Obstacle avoidance trajectories for various initial heading angles in the range $\left( 0, \frac{3 \pi}{4} \right)$. \textbf{(bottom left)} The value of $h_0$, derived from (\ref{eq: planar drone safety specification}), which stays positive for all initial headings. \textbf{(top right)} The value of the DRD-CBF (\ref{eq:our-cbf}). \textbf{(bottom right)} The value of the attitude Lyapunov function $V$ which is used to penalize $h_0$.}
\label{fig:obstacle_avoid planar drone}
\vspace{-6mm}
\end{figure}

\subsection{Planar Quadrotor: Safety-Critical Control}
As in the previous case study, when $\|\tilde \bk (\bx)\| = \|\hat{\bk}(\vec{\by}(\bx)) - L_{\bf}^2\by(\bx))\| \neq 0$, we can define a desired safe heading angle as $\theta_\des(\bx)=  \mathrm{atan2}(\hat \bk_y(\vec{\by}(\bx)) + g, \hat \bk_x(\vec{\by}(\bx)))$ and consider a function as defined in \eqref{eq: our-tracking clf explicit},  denoted by $V_0$, whose derivative is of the form:
 \begin{align}
 \dot V_0 (\bx) = \frac{\partial V_0}{\partial \vec{\by}}(\bx)\dot{\vec{\by}} + \frac{\partial V_0}{\partial \theta}(\bx) \omega.
 \end{align}
 However, since the quadrotor model \eqref{eq: planar drone ex} is a second-order system, we must design a moment $\mathbf{M}$ that drives $\omega $ to a desired state that ensures $\theta$ can be stabilized to $\theta_\des(\bx)$.
To achieve this, we construct a CLF using backstepping \cite{Krstic} as: 
\begin{align}\label{eq: backstepping V}
V(\bx) &= V_0(\bx) + \frac{1}{2 \mu_2}\|\omega   - \bk_\omega(\bx)\|^2,
\end{align}
with $\mu_2 >0$, such that $\bk_\omega:\real^n \rightarrow \real$ satisfies:
\begin{align}\label{eq: k_w}
    \frac{\partial V_0}{\partial \vec{\by}}(\bx)\dot{\vec{\by}} + \frac{\partial V_0}{\partial \theta}(\bx) \bk_\omega(\bx)  \leq -\lambda V_0(\bx).
\end{align}
%
%
One can verify that \eqref{eq: backstepping V} satisfies \eqref{eq:k2} on a set $\mathcal{X}$ for \eqref{eq: planar drone ex}.
We consider the DRD-CBF as in \eqref{eq:our-cbf}, and
show that $h$ is a CBF for \eqref{eq: planar drone ex}.
Using similar analysis to the previous case study, we observe that $L_{\bg{_2}} V =0 
\implies \omega = \bk_\omega(\bx)$ which, from \eqref{eq: k_w}, implies $\dot V = L_{\bf_1} V \leq -\lambda V$, so the Lyapunov condition can be satisfied on all of $\mathcal{X}$. Again, note that we are not claiming that the Lyapunov condition \eqref{eq:k2} may be satisfied by a continuous controller, but that at each state, there exists inputs that render the system  \eqref{eq: planar drone ex} safe, which is sufficient to ensure that $h$ is a CBF as in Def.~\ref{def:cbf}.
\begin{figure}
    \centering
    \includegraphics[width=1\linewidth]{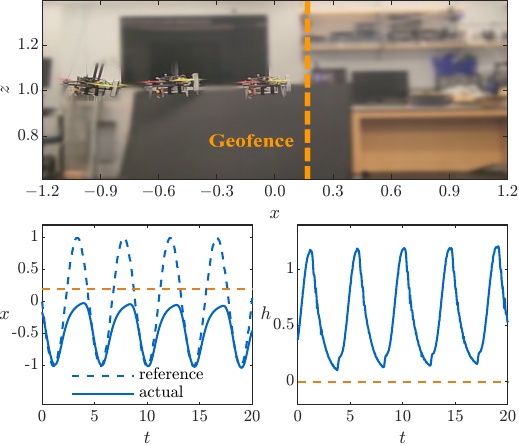}
    \vspace{-7mm}
    \caption{3D quadrotor demonstration. \textbf{(top)} A composite image showing the position of the quadrotor drone over the course of the geofencing experiment. \textbf{(bottom left)} The $x$-position reference, which passes beyond the geofence, and the actual $x$-position, which deviates from the reference to maintain safety. \textbf{(bottom right)} The value of the DRD-CBF (\ref{eq:our-cbf}), which stays positive throughout the flight, confirming that safety is maintained.}
    \label{fig:drone}
    \vspace{-5mm}
\end{figure}
\subsection{Planar Quadrotor: Simulation}
 We implement the above example in simulation with the safety-specification from \eqref{eq: planar drone safety specification} 
with a nominal tracking controller $\bk_\mathrm{nom}(\bx) = [\bk_1(\bx), K_\mathrm{p}(\theta_\des(\bx) - \theta) + K_\mathrm{q}(\dot{\theta}_\des(\bx) - \omega)]^\top$ where $K_\mathrm{p}, K_\mathrm{q} \in \real_{>0}$ are gains.  The resulting DRD-CBF is used to construct \eqref{eq:cbfqp} and
Fig. \ref{fig:obstacle_avoid planar drone} illustrates the trajectories of the planar quadrotor resulting from filtering $\bk_\mathrm{nom}$ with $h$ along with the satisfaction of safety and the tracking behavior encoded through the CLF.
\subsection{ 3D Quadrotor: Hardware}
Finally, we deploy our method on a quadrotor drone. 
We use an OptiTrack motion capture system to provide the drone with real-time position measurements and a VectorNav VN-200 IMU for attitude state estimation. All state estimation and control computations are performed onboard at 750 Hz using a Jetson Orin NX. The drone model used is a simplified version of the dynamics \eqref{eq:3D_drone} with thrust and desired angle rate inputs as in \cite{CosnerICRA24}. The desired angle rates are tracked by a Betaflight flight controller and ESC at 8 kHz.

To demonstrate the performance of our DRD-CBF \eqref{eq:our-cbf}, we command the quadrotor to track a sinusoidal reference $\by_\text{d}(t) = [-\sin{(0.4\pi t)}, \ 0.0, \ 1.0]^\top$. We then define an $x$-coordinate geofence as the 0-superlevel set of $h_\text{geo}(\vec{\by}(\bx)) = x_\text{geo} - x,$ where $x_\text{geo}\in\mathbb{R}$ is the $x$-position of the geofence \cite{drewgeofence}. For this particular experiment $x_\text{geo}=0.2\;\text{m}$. Using a high order CBF \cite{SreenathACC16,WeiTAC22}, we extend $h_\text{geo}$ to get $h_0$, a CBF for the quadrotor double-integrator translational dynamics. By enforcing forward invariance of the safe-set defined by $h_0(\vec{\by}(\bx)) \!\!\geq \!\! 0$, we ensure the $x$-coordinate of the quadrotor never exceeds the value of $x_\text{geo}$, irrespective of the commanded reference. 
Select data are presented in Fig.\,\ref{fig:drone} utilizing the DRD-CBF \eqref{eq:our-cbf} with the tracking CLF \eqref{eq: our-tracking clf explicit} for 3D rotation matrices (see appendix).


From Fig.\,\ref{fig:drone}a and Fig.\,\ref{fig:drone}b, it is clear that the quadrotor drone effectively tracks the sinusoidal reference as long as it stays inside the safe set. However, as the commanded position crosses the geofence, the safety filter intercedes, preventing the drone from violating its safety specification.

%% file: simulations.tex
\section{Conclusion}
We presented a constructive framework for synthesizing CBFs and safety-critical controllers for nonlinear systems with dual relative degree, where outputs are used to specify safety requirements.
%
%
We design a CBF for an integrator chain with a Lyapunov function to certify tracking of the safe inputs generated by this system, to synthesize a CBF for the full nonlinear system.
We also provide case studies of systems of dual relative degree, for which we synthesize CBFs.
We further demonstrated the efficacy of the proposed method on hardware platforms that exhibit these properties. Future work involves studying the impact of uncertainties in high-fidelity scenarios in navigation for this class of systems.

%% file: appdx.tex
\appendix
We show that the tracking CLF \eqref{eq: our-tracking clf explicit} satisfies \eqref{eq:clf-lower-bound} for 3D rotation matrices---the result holds for 2D rotation matrices as a special case and as such, holds for all case studies presented in this paper. 
First, define the \textit{error rotation matrix} $\mathbf{R}_\mathrm{e}$ as the rotation matrix which represents the attitude of the \textit{desired} frame with respect to the \textit{body-fixed} frame, i.e.,
\begin{equation}\mathbf{R}_\mathrm{e} = \mathbf{R}^\top\mathbf{R}_\mathrm{des}.
\end{equation}

This rotation matrix can always been represented using Euler angles, that is, it can be represented by a sequence of three (not necessarily unique) rotations about intrinsic orthogonal axes. In particular, we can choose the Z-X-Z rotation sequence, resulting in:
\begin{equation}
    \mathbf{R}_\mathrm{e} = \mathbf{R}_z(\psi)\mathbf{R}_x(\theta)\mathbf{R}_z(\varphi).
\label{eq: Euler}
\end{equation}

Given $\mathbf{R}_\mathrm{e}$, for any Euler angle rotation sequence, the resultant $\varphi$ and $\psi$ are \textit{not} guaranteed to be unique (specifically when the second Euler angle $\theta$ is defined such that it aligns the 1st and 3rd rotation axes). However, $\theta$ will always be unique.
Next, consider the geometric tracking CLF:
\begin{align}
    V(\bx) &= \frac{\|\tilde{\mathbf{k}}(\bx)\|^2}{2}\mathrm{tr}\left( \mathbf{I} - \mathbf{R}_\mathrm{e}\right),
\label{eq: Lyapunov}
\end{align}

\noindent which is a slight modification of the Lyapunov function in \cite{lee_geometric_2010} used to prove ``near" global asymptotic stability of the equilibrium point defined by $\mathbf{R}_\mathrm{e} = \mathbf{I}$. Futhermore, for the 3D quadrotor \eqref{eq:3D_drone}, we have $L_{\bg_1}L_\bf \by(\bx) = \mathbf{R}\mb{e}_z$, thus, following from \eqref{eq:tracking-error}, we can express $\Vert \mb{e}(\mb{x}) \Vert$ as:
\begin{align}
    \|\mb{e}(\bx)\| = \|(\mathbf{I}-\mathbf{R}\mathbf{e}_z\mathbf{e}_z^\top\mathbf{R}^\top)\tilde \bk(\bx)\|.
\end{align}
\noindent Since $\mathbf{b}_z = \mathbf{R}\mathbf{e}_z$ is the basis vector representing the body-fixed z-axis, this becomes:
\begin{align}
    \|\mb{e}(\bx)\| &= \|(\mathbf{I}-\mathbf{b}_z\mathbf{b}_z^\top)\tilde \bk(\bx)\|\\
    &= \|\tilde \bk(\bx)-\mathbf{b}_z(\mathbf{b}_z^\top\tilde \bk(\bx))\|,
\label{ortho}
\end{align}
where the vector $\tilde{\bk}(\bx)-\mathbf{b}_z(\mathbf{b}_z^\top\tilde \bk(\bx))$ in (\ref{ortho}) is the component of $\tilde \bk(\bx)$ orthogonal to $\mathbf{b}_z$. Since the angle between $\tilde \bk(\bx)$ and $\mathbf{b}_z$ is precisely the $\theta$ in (\ref{eq: Euler}), the magnitude of the projection can be expressed as:
\begin{align}\label{eq: error explicit}
    \|\mb{e}(\bx)\| = \|\tilde \bk(\bx) \sin(\theta)\|.
\end{align}
Using this simplified expression for $\Vert \mb{e}(\mb{x}) \Vert$, we now present the following Lemma.

\begin{lemma} There exists a $\beta >0$ such that the tracking CLF \eqref{eq: Lyapunov} is lower bounded by \eqref{eq: error explicit} as in \eqref{eq:clf-lower-bound} for all $\bx \in \re^n$.
\end{lemma}

\begin{proof}
By substituting (\ref{eq: Euler}) into (\ref{eq: Lyapunov}), and dropping the dependency on $\bx$ for brevity, we arrive at:

\begin{align}
    &V(\bx) = \frac{\|\tilde{\mathbf{k}}\|^2}{2}\mathrm{tr}\left( \mathbf{I} - \mathbf{R}_z(\psi)\mathbf{R}_x(\theta)\mathbf{R}_z(\varphi)\right)\\
    &= \frac{\|\tilde{\mathbf{k}}\|^2}{2}\left(\mathrm{tr}\left( \mathbf{I}\right) - \mathrm{tr}\left(\mathbf{R}_z(\psi)\mathbf{R}_x(\theta)\mathbf{R}_z(\varphi)\right)\right)\\
    &= \frac{\|\tilde{\mathbf{k}}\|^2}{2}\left(3 - \left(\cos(\varphi+\psi) + \cos(\varphi+\psi)\cos(\theta) +\cos(\theta)\right)\right)\\ 
    &= \frac{\|\tilde{\mathbf{k}}\|^2}{2}\left(4 - \left(\cos(\varphi+\psi)(\cos(\theta)+1) +(\cos(\theta) + 1)\right)\right)\\
    &= \|\tilde{\mathbf{k}}\|^2\left(2 - \left(\frac{\cos(\varphi+\psi)+1}{2}\right)(\cos(\theta)+1)\right).
\label{eq: Lyapunov_Euler}
\end{align}

Because $\frac{\cos(\cdot)+1}{2}$ exists on the domain [0,1], and $\cos(\cdot)+1$ exists on [0,2], we can say:

\begin{align}
    V(\bx) &= \|\tilde{\mathbf{k}}\|^2\left(2 - \left(\frac{\cos(\varphi+\psi)+1}{2}\right)(\cos(\theta)+1)\right)\\
    &\geq \|\tilde{\mathbf{k}}\|^2\left(2 - (\cos(\theta)+1)\right)\\
    &= \|\tilde{\mathbf{k}}\|^2\left(1 - \cos(\theta)\right)\\
    &\geq \|\tilde{\mathbf{k}}\|^2\left(1 - \cos(\theta)\right)\left(\frac{\cos(\theta)+1}{2}\right)\\
    &= \|\tilde{\mathbf{k}}\|^2\frac{1 - \cos^2(\theta)}{2}\\
    &= \|\tilde{\mathbf{k}}\|^2\frac{\sin^2(\theta)}{2}\\
    &= \frac{1}{2}\|\tilde{\mathbf{k}}\sin(\theta)\|^2\\
    &= \frac{1}{2}\|\mathbf{e}(\bx)\|^2\\
    &\geq \beta\|\mathbf{e}(\bx)\|^2, \quad \forall \beta \in \left(0,\;\frac{1}{2}\right].
\end{align}

Therefore, the choice of CLF (\ref{eq: Lyapunov}) satisfies \eqref{eq:clf-lower-bound}. 
\end{proof}